\documentclass[runningheads]{templates/llncs}

\usepackage{STYs/DPMC}
\usepackage{STYs/tensors}

\newcommand{\doi}[1]{\textsc{doi}: \href{http://dx.doi.org/#1}{\nolinkurl{#1}}}

\begin{document}


\title{
    DPMC: Weighted Model Counting by Dynamic Programming on Project-Join Trees%
    \thanks{Work supported in part by NSF grants IIS-1527668, CCF-1704883, IIS-1830549, and DMS-1547433.}
}
\titlerunning{ 
    Weighted Model Counting by Dynamic Programming on Project-Join Trees
}

\author{
    Jeffrey M. Dudek
    \and Vu H. N. Phan
    \and Moshe Y. Vardi
}
\authorrunning{J. M. Dudek et al.}

\institute{
    Rice University, Houston TX 77005, USA \\
    \email{\{jmd11,vhp1,vardi\}@rice.edu}
}

\maketitle



\begin{abstract}
    We propose a unifying dynamic-programming framework to compute exact literal-weighted model counts of formulas in conjunctive normal form.
    At the center of our framework are project-join trees, which specify efficient project-join orders to apply additive projections (variable eliminations) and joins (clause multiplications).
    In this framework, model counting is performed in two phases.
    First, the planning phase constructs a project-join tree from a formula.
    Second, the execution phase computes the model count of the formula, employing dynamic programming as guided by the project-join tree.
    We empirically evaluate various methods for the planning phase and compare constraint-satisfaction heuristics with tree-decomposition tools.
    We also investigate the performance of different data structures for the execution phase and compare algebraic decision diagrams with tensors.
    We show that our dynamic-programming model-counting framework \Dpmc{} is competitive with the state-of-the-art exact weighted model counters \cachet, \ctd, \df, and \minictd.
\end{abstract}


\section{Introduction}
\label{sec_intro}
\textdef{Model counting} is a fundamental problem in artificial intelligence, with applications in machine learning, probabilistic reasoning, and verification \cite{domshlak2007probabilistic,gomes2009model,naveh2007constraint}.
Given an input set of constraints, with the focus in this paper on Boolean constraints, the model-counting problem is to count the number of satisfying assignments.
Although this problem is \#P-Complete \cite{valiant1979complexity}, a variety of tools exist that can handle industrial sets of constraints, \eg, \cite{sang2004combining,oztok2015top,darwiche2004new,lagniez2017improved}.

Dynamic programming is a powerful technique that has been applied across computer science \cite{bellman1966dynamic}, including to model counting \cite{bacchus2009solving,samer2010algorithms,jegou2016improving}.
The key idea is to solve a large problem by solving a sequence of smaller subproblems and then incrementally combining these solutions into the final result.
Dynamic programming provides a natural framework to solve a variety of problems defined on sets of constraints: subproblems can be formed by partitioning the constraints.
This framework has been instantiated into algorithms for database-query optimization \cite{mcmahan2004projection}, satisfiability solving \cite{uribe1994ordered,aguirre2001random,pan2005symbolic}, and QBF evaluation \cite{charwat2016bdd}.

Dynamic programming has also been the basis of several tools for model counting \cite{dudek2020addmc,dudek2019efficient,dudek2020parallel,fichte2020exploiting}.
Although each tool uses a different data structure--algebraic decision diagrams (ADDs) \cite{dudek2020addmc}, tensors \cite{dudek2019efficient,dudek2020parallel}, or database tables \cite{fichte2020exploiting}--the overall algorithms have similar structure.
The goal of this work is to unify these approaches into a single conceptual framework: \emph{project-join trees}.
Project-join trees are not an entirely new idea.
Similar concepts have been used in constraint programming (as join trees \cite{dechter1989tree}), probabilistic inference (as cluster trees \cite{shachter1994global}), and database-query optimization (as join-expression trees \cite{mcmahan2004projection}).
Our original contributions include the unification of these concepts into project-join trees and the application of this unifying framework to model counting.

We argue that project-join trees provide the natural formalism to describe execution plans for dynamic-programming algorithms for model counting.
In particular, considering project-join trees as \emph{execution plans} enables us to decompose dynamic-programming algorithms such as the one in \cite{dudek2020addmc} into two phases, following the breakdown in \cite{dudek2020parallel}: a \emph{planning} phase and an \emph{execution} phase.
This enables us to study and compare different planning algorithms, different execution environments, and the interplay between planning and execution.
Such a study is the main focus of this work.
While the focus here is on model counting, our framework is of broader interest.
For example, in \cite{tabajara2017factored}, Tabajara and Vardi described a dynamic-programming, binary-decision-diagram-based framework for functional Boolean synthesis.
Refactoring the algorithm into a planning phase followed by an execution phase is also of interest in that context.

The primary contribution of the work here is a dynamic-programming framework for weighted model counting based on project-join trees.
In particular:
\begin{enumerate}
    \item We show that several recent algorithms for weighted model counting \cite{dudek2020addmc,dudek2019efficient,fichte2020exploiting} can be unified into a single framework using project-join trees.
    \item We compare the one-shot%
    \footnote{A \emph{one-shot} algorithm outputs exactly one solution and then terminates immediately.}
    constraint-satisfaction heuristics used in \cite{dudek2020addmc} with the anytime%
    \footnote{An \emph{anytime} algorithm outputs better and better solutions the longer it runs.} tree-decomposition tools used in \cite{dudek2019efficient} and observe that
    tree-decomposition tools outperform
    constraint-satisfaction heuristics.
    \item We compare (sparse) ADDs \cite{bahar1997algebraic} with (dense) tensors \cite{kjolstad2017tensor} and find that ADDs outperform tensors on single CPU cores.
    \item We find that project-join-tree-based algorithms contribute to a portfolio of model counters containing \cachet{} \cite{sang2004combining}, \ctd{} \cite{darwiche2004new}, \df{} \cite{lagniez2017improved}, and \minictd{} \cite{oztok2015top}.
\end{enumerate}
These conclusions have significance beyond model counting.
The superiority of anytime tree-decomposition tools over classical one-shot constraint-satisfaction heuristics can have broad applicability.
Similarly, the advantage of compact data structures for dynamic programming may apply to other optimization problems.


\section{Preliminaries}
\label{sec_prelim}


\paragraph{\textbf{Pseudo-Boolean Functions and Early Projection}}

A \emph{pseudo-Boolean function} over a set $X$ of variables is a function $f: 2^X \to \R$.
Operations on pseudo-Boolean functions include \emph{product} and \emph{projection}.
First, we define product.
\begin{definition}[Product]
\label{def_mult}
    Let $X$ and $Y$ be sets of Boolean variables.
    The \textdef{product} of functions $f: 2^X \to \R$ and $g: 2^Y \to \R$ is the function $f \mult g: 2^{X \cup Y} \to \R$ defined for all $\tau \in 2^{X \cup Y}$ by
    $(f \mult g)(\tau) \equiv f(\tau \cap X) \mult g(\tau \cap Y).$
\end{definition}

Next, we define (additive) projection, which marginalizes a single variable.
\begin{definition}[Projection]
\label{def_proj}
    Let $X$ be a set of Boolean variables and $x \in X$.
    The \textdef{projection} of a function $f: 2^X \to \R$ \wrt{} $x$ is the function $\proj_x f: 2^{X \setminus \set{x}} \to \R$ defined for all $\tau \in 2^{X \setminus \set{x}}$ by
    $\pars{\proj_x f}(\tau) \equiv f(\tau) + f(\tau \cup \set{x}).$
\end{definition}
Note that projection is commutative, \ie, that $\proj_x \proj_y f = \proj_y \proj_x f$ for all variables $x, y \in X$ and functions $f: 2^X \to \R$.
Given a set $X = \set{x_1, x_2, \ldots, x_n}$, define
$\proj_X f \equiv \proj_{x_1} \proj_{x_2} \ldots \proj_{x_n} f$.
Our convention is that $\proj_\emptyset f \equiv f$.

When performing a product followed by a projection, it is sometimes possible to perform the projection first.
This is known as \emph{early projection} \cite{mcmahan2004projection}.
\begin{theorem}[Early Projection]
\label{thm_early_proj}
    Let $X$ and $Y$ be sets of variables.
    For all functions $f: 2^X \to \R$ and $g: 2^Y \to \R$, if $x \in X \setminus Y$, then $\proj_x (f \mult g) = \pars{\proj_x f} \mult g.$
\end{theorem}
Early projection is a key technique in symbolic computation in a variety of settings, including database-query optimization \cite{kolaitis2000conjunctive}, symbolic model checking \cite{burch1991symbolic}, satisfiability solving \cite{pan2005symbolic}, and model counting \cite{dudek2020addmc}.


\paragraph{\textbf{Weighted Model Counting}}

We compute the total weight, subject to a given weight function, of all models of an input propositional formula.
Formally:
\begin{definition}[Weighted Model Count]
    Let $X$ be a set of Boolean variables, $\phi: 2^X \to \{0,1\}$ be a Boolean function, and $W: 2^X \to \R$ be a pseudo-Boolean function.
    The \textdef{weighted model count} of $\phi$ \wrt{} $W$ is
    $W(\phi) \equiv \sum_{\tau \in 2^X} \phi(\tau) \mult W(\tau)$.
\end{definition}

The weighted model count of $\phi$ \wrt{} $W$ can be naturally expressed in terms of pseudo-Boolean functions: $W(\phi) = \left(\sum_X (\phi \mult W) \right)(\emptyset)$.
The function $W: 2^X \to \R$ is called a \textdef{weight function}.
In this work, we focus on {literal-weight functions}, which can be expressed as products of weights associated with each variable.
Formally, a \textdef{literal-weight function} $W$ can be factored as $W = \prod_{x \in X} W_x$ for pseudo-Boolean functions $W_x: 2^{\{x\}} \to \R$.


\paragraph{\textbf{Graphs}}

A \emph{graph} $G$ has a
set $\V{G}$ of vertices, a set $\E{G}$ of (undirected) edges, a function $\delta_G: \V{G} \to 2^{\E{G}}$ that gives the set of edges incident to each vertex, and a function $\epsilon_G: \E{G} \to 2^{\V{G}}$ that gives the set of vertices incident to each edge.
Each edge must be incident to exactly two vertices, but multiple edges can exist between two vertices.
A \emph{tree} is a simple, connected, and acyclic graph.
A \emph{leaf} of a tree $T$ is a vertex of degree one, and we use $\Lv{T}$ to denote the set of leaves of $T$.
We often refer to a vertex of a tree as a \emph{node} and an edge as an \emph{arc} to avoid confusion.
A \emph{rooted tree} is a tree $T$ together with a distinguished node $r \in \V{T}$ called the \emph{root}.
In a rooted tree $(T, r)$, each node $n \in \V{T}$ has a (possibly empty) set of \emph{children}, denoted $\C(n)$, which contains all nodes $n'$ adjacent to $n$ \st{} all paths from $n'$ to $r$ contain $n$.

\section{Using Project-Join Trees for Weighted Model Counting}
\label{sec_jointree}

In model counting, a Boolean formula is often given in conjunctive normal form (CNF), \ie, as a set $\phi$ of clauses.
For each clause $c \in \phi$, define $\vars(c)$ to be the set of variables appearing in $c$.
Then $c$ represents a Boolean function over $\vars(c)$. Similarly, $\phi$ represents a Boolean function over $\vars(\phi) \equiv \bigcup_{c \in \phi} \vars(c)$.

It is well-known that weighted model counting can be performed through a sequence of projections and joins on pseudo-Boolean functions \cite{dudek2020addmc,dudek2019efficient}.
Given a CNF formula $\phi$ and a literal-weight function $W$ over a set $X$ of variables, the corresponding weighted model count can be computed as follows:
\begin{equation}
\label{eq_factored_wmc}
    W(\phi) = \pars{
        \proj_X
        \pars{\prod_{c \in \phi} c \mult \prod_{x \in X} W_x}
    }(\emptyset)
\end{equation}

By taking advantage of the associative and commutative properties of multiplication as well as the commutative property of projection, we can rearrange Equation \eqref{eq_factored_wmc} to apply early projection.
It was shown in \cite{dudek2020addmc} that early projection can significantly reduce computational cost.
There are a variety of possible rearrangements of Equation \eqref{eq_factored_wmc} of varying costs.
Although \cite{dudek2020addmc} considered several heuristics for performing this rearrangement (using bucket elimination \cite{dechter1999bucket} and Bouquet's Method \cite{bouquet1999gestion}), they did not attempt to analyze rearrangements.

In this work, we aim to analyze the quality of the rearrangement, in isolation from the underlying implementation and data structure used for Equation \eqref{eq_factored_wmc}.
This approach has been highly successful for database-query optimization \cite{mcmahan2004projection}, where the central object of theoretical reasoning is the \emph{query plan}.
The approach has also seen similar success in Bayesian network inference \cite{darwiche1998dynamic}.

We model a rearrangement of Equation \eqref{eq_factored_wmc} as a \emph{project-join tree}:
\begin{definition}[Project-Join Tree]
\label{def_jointree}
    Let $X$ be a set of Boolean variables and $\phi$ be a CNF formula over $X$.
    A \emph{project-join tree} of $\phi$ is a tuple $(T, r, \gamma, \pi)$ where:
    \begin{itemize}
        \item $T$ is a tree with root $r \in \V{T}$,
        \item $\gamma: \Lv{T} \to \phi$ is a bijection between the leaves of $T$ and the clauses of $\phi$, and
        \item $\pi: \V{T} \setminus \Lv{T} \to 2^X$ is a labeling function on internal nodes.
    \end{itemize}
    Moreover, $(T, r, \gamma, \pi)$ must satisfy the following two properties:
    \begin{enumerate}[ref=\arabic*]
        \item $\{\pi(n) : n \in \V{T} \setminus \Lv{T} \}$ is a partition of $X$, and \label{prop1}
        \item for each internal node $n \in \V{T} \setminus \Lv{T}$, variable $x \in \pi(n)$, and clause $c \in \phi$ \st{} $x$ appears in $c$, the leaf node $\gamma^{-1}(c)$ must be a descendant of $n$ in $T$. \label{prop2}
    \end{enumerate}
\end{definition}
If $n$ is a leaf node, then $n$ corresponds to a clause $c = \gamma(n)$ in Equation \eqref{eq_factored_wmc}.
If $n$ is an internal node, then $n$'s children $\C(n)$ are to be multiplied before the projections of variables in $\pi(n)$ are performed.
The two properties ensure that the resulting expression is equivalent to Equation \eqref{eq_factored_wmc} using early projection.
See Figure \ref{fig_join_tree} for a graphical example of a project-join tree.
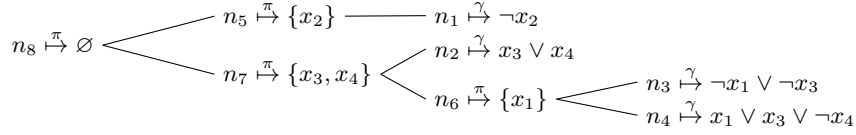
\begin{figure}
    \centering
    \begin{tikzpicture}[grow=right] 
        \tikzset{level distance=80pt,sibling distance=-6pt}
        \tikzset{execute at begin node=\strut}
        \tikzset{every tree node/.style={anchor=base west}}
        \Tree [ .$n_{8}\piMap\emptyset$
            [ .$n_{7}\piMap\set{x_3, x_4}$
                [ .$n_{6}\piMap\set{x_1}$
                    [ .$n_4\gammaMap{x_1 \vee x_3 \vee \neg x_4}$ ]
                    [ .$n_3\gammaMap{\neg x_1 \vee \neg x_3}$ ]
                ]
                [ .$n_2\gammaMap{x_3 \vee x_4}$ ]
            ]
            [ .$n_{5}\piMap\set{x_2}$ [ .$n_1\gammaMap{\neg x_2}$ ] ]
        ]
    \end{tikzpicture}
\caption{
    A project-join tree $(T, n_{8}, \gamma, \pi)$ of a CNF formula $\phi$.
    Each leaf node is labeled by $\gamma$ with a clause of $\phi$.
    Each internal node is labeled by $\pi$ with a set of variables of $\phi$.
}
\label{fig_join_tree}
\end{figure}

Project-join trees have previously been studied in the context of database-query optimization \cite{mcmahan2004projection}.
Project-join trees are closely related to contraction trees in the context of tensor networks \cite{evenbly2014improving,dudek2019efficient}.
Once a rearrangement of Equation \eqref{eq_factored_wmc} has been represented by a project-join tree, we can model the computation process according to the rearrangement.
In particular, given a literal-weight function $W = \prod_{x \in X} W_x$, we define the $W$-\emph{valuation} of each node $n \in \V{T}$ as a pseudo-Boolean function associated with $n$.
The $W$-valuation of a node $n \in \V{T}$ is denoted $f^W_n$ and defined as follows:
\begin{equation}
\label{eq_valuation}
    f^W_n \equiv
    \begin{cases}
       \gamma(n) & \text{if}~n \in \Lv{T} \\
        \sum_{\pi(n)} \pars{ \prod_{o \in \C(n)} f^W_o \cdot \prod_{x \in \pi(n)} W_x } & \text{if}~n \notin \Lv{T}
    \end{cases}
\end{equation}

Note that the $W$-valuation of a leaf node $n \in \Lv{T}$ is a clause $c = \gamma(n) \in \phi$, interpreted in this context as an associated function $\lambda_c : 2^{\vars(c)} \to \B$ where $\lambda_c(\tau) = 1$ if and only if the truth assignment $\tau$ satisfies $c$.
The main idea is that the $W$-valuation at each node of $T$ is a pseudo-Boolean function computed as a subexpression of Equation \eqref{eq_factored_wmc}.
The $W$-valuation of the root is exactly the result of Equation \eqref{eq_factored_wmc}, \ie, the weighted model count of $\phi$ \wrt{} $W$:
\begin{theorem}
\label{thm_valuation_wmc}
    Let $\phi$ be a CNF formula over a set $X$ of variables, $(T, r, \gamma, \pi)$ be a project-join tree of $\phi$, and $W$ be a literal-weight function over $X$.
    Then $f^W_r(\emptyset) = W(\phi)$.
\end{theorem}

This gives us a two-phase algorithm for computing the weighted model count of a formula $\phi$.
First, in the \emph{planning} phase, we construct a project-join tree $(T, r, \gamma, \pi)$ of $\phi$.
We discuss algorithms for constructing project-join trees in Section \ref{sec_planning}.
Second, in the \emph{execution} phase, we compute $f^W_r$ by following Equation \eqref{eq_valuation}.
We discuss data structures for computing Equation \eqref{eq_valuation} in Section \ref{sec_execution}.

When computing a $W$-valuation, the number of variables that appear in each intermediate pseudo-Boolean function has a significant impact on the runtime.
The set of variables that appear in the $W$-valuation of a node is actually independent of $W$.
In particular, for each node $n \in \V{T}$, define $\vars(n)$ as follows:
\begin{equation}
    \vars(n) \equiv
    \begin{cases}
        \vars(\gamma(n)) & \text{if}~n \in \Lv{T} \\
        \pars{\bigcup_{o \in \C(n)} \vars(o)} \setminus \pi(n) & \text{if}~n \notin \Lv{T}
    \end{cases}
\end{equation}

For every literal-weight function $W$, the domain of the function $f^W_n$ is $2^{\vars(n)}$.
To characterize the difficulty of $W$-valuation, we define the \emph{size} of a node $n$, $\func{size}(n)$, to be $\size{\vars(n)}$ for leaf nodes and $\size{\vars(n) \cup \pi(n)}$ for internal nodes.
The \emph{width} of a project-join tree $(T, r, \gamma, \pi)$ is $\func{width}(T) \equiv \max_{n \in \V T} \func{size}(n)$.
We see in Section \ref{sec_experiments} how the width impacts the computation of $W$-valuations.

\section{Planning Phase: Building a Project-Join Tree}
\label{sec_planning}

In the planning phase, we are given a CNF formula $\phi$ over Boolean variables $X$.
The goal is to construct a project-join tree of $\phi$.
In this section, we present two classes of techniques that have been applied to model counting:
using constraint-satisfaction heuristics (in \cite{dudek2020addmc}) and using tree decompositions (in \cite{dudek2019efficient,fichte2020exploiting}).


\subsection{Planning with One-Shot Constraint-Satisfaction Heuristics}
\label{sec_csp}

A variety of constraint-satisfaction heuristics for model counting were presented in a single algorithmic framework by \cite{dudek2020addmc}.
These heuristics have a long history in constraint programming \cite{dechter2003constraint}, database-query optimization \cite{mcmahan2004projection}, and propositional reasoning \cite{pan2005symbolic}.
In this section, we adapt the framework of \cite{dudek2020addmc} to produce project-join trees.
This algorithm is presented as Algorithm \ref{alg_csp_jt}, which constructs a project-join tree of a CNF formula using constraint-satisfaction heuristics.
The functions $\clusterVarOrder$, $\clauseRank$, and $\chosenCluster$ represent heuristics for fine-tuning the specifics of the algorithm.
Before discussing the various heuristics, we assert the correctness of Algorithm \ref{alg_csp_jt} in the following theorem.
\begin{theorem}
\label{thm_csp_jt}
    Let $X$ be a set of variables and $\phi$ be a CNF formula over $X$.
    Assume that $\clusterVarOrder$ returns an injection $X \to \N$.
    Furthermore, assume that all $\clauseRank$ and $\chosenCluster$ calls satisfy the following conditions:
    \begin{enumerate}[ref=\arabic*]
        \item $1 \le \clauseRank(c, \rho) \le m$, \label{cond1}
        \item $i < \chosenCluster(n_i) \le m$, and \label{cond2}
        \item $X_s \cap \vars(n_i) = \emptyset$ for all integers $s$ where $i < s < \chosenCluster(n_i)$. \label{cond3}
    \end{enumerate}
    Then Algorithm \ref{alg_csp_jt} returns a project-join tree of $\phi$.
\end{theorem}

\begin{algorithm*}[t]
\label{alg_csp_jt}
\caption{Using combined constraint-satisfaction heuristics to build a project-join tree}
    \DontPrintSemicolon
    \KwIn{$X$: set of $m \ge 1$ Boolean variables}
    \KwIn{$\phi$: CNF formula over $X$}
    \KwOut{$(T, r, \gamma, \pi)$: project-join tree of $\phi$}
    $(T, \nil, \gamma, \pi) \gets \text{empty project-join tree}$\;
    $\rho \gets \clusterVarOrder(\phi)$
        \tcc*{injection $\rho : X \to \N$}
    \For{$i = m, m - 1, \ldots, 1$}{
        $\Gamma_i \gets \set{c \in \phi : \clauseRank(c, \rho) = i}$
            \tcc*{$1 \le \clauseRank(c, \rho) \le m$}
        $\kappa_i \gets \set{\leaf(T, c) : c \in \Gamma_i}$\;
            \tcc*{for each $c$, a leaf $l$ with $\gamma(l) = c$ is constructed and put in cluster $\kappa_i$}
        $X_i \gets \vars(\Gamma_i) \setminus \bigcup_{j = i + 1}^m \vars(\Gamma_j)$
            \tcc*{$\set{X_i}_{i = 1}^m$ is a partition of $X$}
    }
    \For{$i = 1, 2, \ldots, m$}{
        \If{$\kappa_i \ne \emptyset$}{
            $n_i \gets \internal(T, \kappa_i, X_i)$
                \tcc*{$\C(n_i) = \kappa_i$ and $\pi(n_i) = X_i$}
                \label{line_internal_node}
            \If{$i < m$}{
                $j \gets \chosenCluster(n_i)$ \label{line_chosen_cluster}
                    \tcc*{$i < j \le m$}
                $\kappa_j \gets \kappa_j \cup \set{n_i}$ \label{line_cluster_union}
            }
        }
    }
    \Return{$(T, n_m, \gamma, \pi)$}
\end{algorithm*}

By Condition \ref{cond1}, we know that $\set{\Gamma_i}_{i = 1}^m$ is a partition of the clauses of $\phi$.
Condition \ref{cond2} ensures that Lines \ref{line_chosen_cluster}-\ref{line_cluster_union} place a new internal node $n_i$ in a cluster that has not yet been processed.
Also on Lines \ref{line_chosen_cluster}-\ref{line_cluster_union}, Condition \ref{cond3} prevents the node $n_i$ from skipping a cluster $\kappa_s$ if there exists some $x \in X_s \cap \vars(n_i)$, since $x$ is projected in iteration $s$, \ie, $x$ is added to $\pi(n_s)$.
These invariants are sufficient to prove that Algorithm \ref{alg_csp_jt} indeed returns a project-join tree of $\phi$.
All heuristics we use in this work satisfy the conditions of Theorem \ref{thm_csp_jt}.

There are a variety of heuristics to fine-tune Algorithm \ref{alg_csp_jt}.
For the function $\clusterVarOrder$, we consider the heuristics \Random, \Mcs{} (\textdef{maximum-cardinality search} \cite{tarjan1984simple}), \Lexp/\Lexm{} (\textdef{lexicographic search for perfect/minimal orders} \cite{koster2001treewidth}), and \Minfill{} (\textdef{minimal fill-in} \cite{dechter2003constraint}) as well as their inverses (\Invmcs, \Invlexp, \Invlexm, and \Invminfill).
Heuristics for $\clauseRank$ include \Be{} (\textdef{bucket elimination} \cite{dechter1999bucket}) and \Bm{} (\textdef{Bouquet's Method} \cite{bouquet1999gestion}).
For $\chosenCluster$, the heuristics we use are \ListH{} and \TreeH{} \cite{dudek2020addmc}.
We combine $\clauseRank$ and $\chosenCluster$ as \textdef{clustering heuristics}: $\Be-\ListH$, $\Be-\TreeH$, $\Bm-\ListH$, and $\Bm-\TreeH$.
These heuristics are described in \cite{dudek2020addmc}%
.


\subsection{Planning with Anytime Tree-Decomposition Tools}
\label{sec_td}

In join-query optimization, \emph{tree decompositions} can be used to compute join trees \cite{dalmau2002constraint,mcmahan2004projection}.
Tree decompositions \cite{robertson1991graph} decompose graphs into tree structures.
\begin{definition}[Tree Decomposition]
	A \emph{tree decomposition} $(S, \chi)$ of a graph $G$ is a tree $S$ with a labeling function $\chi : \V{S} \to 2^{\V{G}}$ where:
	\begin{enumerate}[ref=\arabic*]
		\item for all $v \in \V{G}$, there exists $n \in \V{S}$ \st{} $v \in \chi(n)$,
		\item for all $e \in \E{G}$, there exists $n \in \V{S}$ \st{} $\einc{G}{e} \subseteq \chi(n)$, and
		\item for all $n, o, p \in \V{S}$, if $o$ is on the path from $n$ to $p$, then $\chi(n) \cap \chi(p) \subseteq \chi(o)$. \label{prop_running_intersection}
	\end{enumerate}
	The \emph{treewidth}, or simply \emph{width}, of $(S, \chi)$ is $\func{tw}(S, \chi) \equiv \max_{n \in \V{S}} \size{\chi(n)} - 1.$
\end{definition}

In particular, join-query optimization uses tree decompositions of the \emph{join graph} to find optimal join trees \cite{dalmau2002constraint,mcmahan2004projection}.
The \emph{join graph} of a project-join query consists of all attributes of a database as vertices and all tables as cliques.
In this approach, tree decompositions of the join graph of a query are used to find optimal project-join trees; see Algorithm 3 of \cite{mcmahan2004projection}.
Similarly, tree decompositions of the \emph{primal graph} of a factor graph, which consists of all variables as vertices and all factors as cliques, can be used to find variable elimination orders \cite{kask2005unifying}.
This technique has also been applied in the context of tensor networks \cite{morgenstern2008ltl,dudek2019efficient}.

Translated to model counting, this technique allows us to use tree decompositions of the \textdef{Gaifman graph} of a CNF formula to compute project-join trees.
The Gaifman graph of a CNF formula $\phi$, denoted $\gaifman(\phi)$, has a vertex for each variable of $\phi$, and two vertices are adjacent if the corresponding variables appear together in some clause of $\phi$.
We present this tree-decomposition-based technique as Algorithm \ref{alg_td_to_join}.
The key idea is that each clause $c$ of $\phi$ forms a clique in $\gaifman(\phi)$ between the variables of $c$.
Thus all variables of $c$ must appear together in some label of the tree decomposition.
We identify that node with $c$.
\begin{algorithm*}[t]
\label{alg_td_to_join}
\caption{Using a tree decomposition to build a project-join tree}
    \DontPrintSemicolon
    \KwIn{$X$: set of Boolean variables}
    \KwIn{$\phi$: CNF formula over $X$}
    \KwIn{$(S, \chi)$: tree decomposition of the Gaifman graph of $\phi$}
    \KwOut{$(T, r, \gamma, \pi)$: project-join tree of $\phi$}
    $(T, \nil, \gamma, \pi) \gets \text{empty project-join tree}$\;
    $found \gets \emptyset$\tcc*{clauses of $\phi$ that have been added to $T$}
    \Function{\upshape $\func{Process}(n, \ell)$}{
        \KwIn{$n \in \V{S}$: node of $S$ to process}
        \KwIn{$\ell \subseteq X$: variables that must not be projected out here}
        \KwOut{$N \subseteq \V{T}$}
        $clauses \gets \{ c \in \phi : c \notin found~\text{and}~\vars(c) \subseteq \chi(n)\}$\; \label{line_clauses}
        $found \gets found \cup clauses$\; \label{line_found}
        $children \gets \{\leaf(T, c) : c \in clauses\} \cup \bigcup_{o \in \C(n)} \func{Process}(o, \chi(n))$\; \label{line_recur}
            \tcc*{new leaf nodes $p \in \V T$ with $\gamma(p) = c$}
        \If{$children = \emptyset$~\text{\upshape or}~$\chi(n) \subseteq \ell$}{
            \Return{children}
        }
        \Else{
            \Return{\upshape $\{\internal(T, children, \chi(n) \setminus \ell)\}$}\; \label{line_return_singleton}
                \tcc*{new internal node $o \in \V T$ with label $\pi(o) = \chi(n) \setminus \ell$}
        }
    }
    $s \gets$ arbitrary node of $S$ \label{line_arbitrary_node}
        \tcc*{fixing $s$ as root of $S$}
    $r \gets \text{only element of}~\func{Process}(s, \emptyset)$\;
    \Return{$(T, r, \gamma, \pi)$}
\end{algorithm*}

The width of the resulting project-join tree is closely connected to the width of the original tree decomposition.
We formalize this in the following theorem.
\begin{theorem}
\label{thm_td_to_join}
	Let $\phi$ be a CNF formula over a set $X$ of variables and $(S, \chi)$ be a tree decomposition of $\gaifman(\phi)$ of width $w$.
    Then Algorithm \ref{alg_td_to_join} returns a project-join tree of $\phi$ of width at most $w+1$.
\end{theorem}
The key idea is that, for each node $n \in \V{S}$, the label $\chi(n)$ is a bound on the variables that appear in all nodes returned by $\func{Process}(n, \cdot)$.
Theorem \ref{thm_td_to_join} allows us to leverage state-of-the-art anytime tools for finding tree decompositions \cite{tamaki2019positive,strasser2017computing,abseher2017htd} to construct project-join trees, which we do in Section \ref{sec_experiments_planning}.

On the theoretical front, it is well-known that tree decompositions of the Gaifman graph are actually equivalent to project-join trees \cite{mcmahan2004projection}.
That is, one can go in the other direction as well: given a project-join tree of $\phi$, one can construct a tree decomposition of $\gaifman(\phi)$ of equivalent width.
Formally:
\begin{theorem}
\label{thm_join_to_td}
    Let $\phi$ be a CNF formula and $(T, r, \gamma, \pi)$ be a project-join tree of $\phi$ of width $w$.
    Then there is a tree decomposition of $\gaifman(\phi)$ of width $w-1$.
\end{theorem}
Theorem \ref{thm_join_to_td} is Lemma 1 of \cite{mcmahan2004projection} and can be seen as the inverse of Theorem \ref{thm_td_to_join}.

\section{Execution Phase: Performing the Valuation}
\label{sec_execution}

The execution phase involves a CNF formula $\phi$ over variables $X$, a project-join tree $(T, r, \gamma, \pi)$ of $\phi$, and a literal-weight function $W$ over $X$.
The goal is to compute the valuation $f^W_r$ using Equation \eqref{eq_valuation}.
Several data structures can be used for the pseudo-Boolean functions that occur while using Equation \eqref{eq_valuation}.
In this work, we consider two data structures that have been applied to weighted model counting: ADDs (as in \cite{dudek2020addmc}) and tensors (as in \cite{dudek2019efficient}).



\subsection{Algebraic Decision Diagrams}

An \emph{algebraic decision diagram (ADD)} is a compact representation of a pseudo-Boolean function as a directed acyclic graph \cite{bahar1997algebraic}.
For functions with logical structure, an ADD representation can be exponentially smaller than the explicit representation.
Originally designed for matrix multiplication and shortest path algorithms, ADDs have also been used for Bayesian inference \cite{chavira2007compiling,gogate2011approximation}, stochastic planning \cite{hoey1999spudd}, model checking \cite{kwiatkowska2007stochastic}, and model counting \cite{fargier2014knowledge,dudek2020addmc}.

Formally, an ADD is a tuple $(X, S, \sigma, G)$, where $X$ is a set of Boolean variables, $S$ is an arbitrary set (called the \textdef{carrier set}), $\sigma: X \to \N$ is an injection (called the \textdef{diagram variable order}), and $G$ is a rooted directed acyclic graph satisfying the following three properties.
First, every leaf node of $G$ is labeled with an element of $S$.
Second, every internal node of $G$ is labeled with an element of $X$ and has two outgoing edges, labeled 0 and 1.
Finally, for every path in $G$, the labels of internal nodes must occur in increasing order under $\sigma$.
In this work, we only need to consider ADDs with the carrier set $S = \mathbb{R}$.

An ADD $(X, S, \sigma, G)$ is a compact representation of a function $f: 2^X \to S$.
Although there are many ADDs representing $f$, for each injection $\sigma: X \to \N$, there is a unique minimal ADD that represents $f$ with $\sigma$ as the diagram variable order, called the \textdef{canonical ADD}.
ADDs can be minimized in polynomial time, so it is typical to only work with canonical ADDs.

Several packages exist for efficiently manipulating ADDs.
For example, \cudd{} \cite{somenzi2015cudd} implements both product and projection on ADDs in polynomial time (in the size of the ADD representation).
\cudd{} was used as the primary data structure for weighted model counting in \cite{dudek2020addmc}.
In this work, we also use ADDs with \cudd{} to compute $W$-valuations.

\Mcs{} was the best diagram variable order on a set of standard weighted model counting benchmarks in \cite{dudek2020addmc}.
So we use \Mcs{} as the diagram variable order in this work.
Note that all other heuristics discussed in Section \ref{sec_csp} for cluster variable order could also be used as heuristics for diagram variable order.

\subsection{Tensors}

A \emph{tensor} is a multi-dimensional generalization of a matrix.
Tensor are widely used in data analysis \cite{cichocki2014era}, signal and image processing \cite{cichocki2015tensor}, quantum physics \cite{arad2010quantum}, quantum chemistry \cite{smilde2005multi}, and many other areas of science.
Given the diverse applications of tensors and tensor networks, a variety of tools \cite{baumgartner2005synthesis,kjolstad2017tensor} exist to manipulate them efficiently on a variety of hardware architectures, including multi-core and GPU-enhanced architectures.

Tensors can be used to represent pseudo-Boolean functions in a dense way.
Tensors are particularly efficient at computing the contraction of two pseudo-Boolean functions: given two functions $f: 2^X \to \mathbb{R}$ and $g: 2^Y \to \mathbb{R}$, their \emph{contraction} $f \contract g$ is the pseudo-Boolean function $\proj_{X \cap Y} f \cdot g$.
The contraction of two tensors can be implemented as matrix multiplication and so leverage significant work in high-performance computing on matrix multiplication on CPUs \cite{lawson1979basic} and GPUs \cite{fatahalian2004understanding}.
To efficiently use tensors to compute $W$-valuations, we follow \cite{dudek2019efficient} in implementing projection and product using tensor contraction.

First, we must compute the weighted projection of a function $f: 2^X \to \mathbb{R}$, \ie, we must compute $\proj_x f \cdot W_x$ for some $x \in X$.
This is exactly equivalent to $f \contract W_x$.
Second, we must compute the product of two functions $f: 2^X \to \mathbb{R}$ and $g: 2^Y \to \mathbb{R}$.
The central challenge is that tensor contraction implicitly projects all variables in $X \cap Y$, but we often need to maintain some shared variables in the result of $f \cdot g$.
In \cite{dudek2019efficient}, this problem was solved using a reduction to tensor networks.
After the reduction, all variables appear exactly twice, so one never needs to perform a product without also projecting all shared variables.

In order to incorporate tensors in our project-join-tree-based framework, we take a different strategy that uses copy tensors.
The \emph{copy tensor} for a set $X$ represents the pseudo-Boolean function $\blacksquare_X: 2^X \to \mathbb{R}$ \st{} $\blacksquare_X(\tau)$ is $1$ if $\tau \in \{ \emptyset, X \}$ and $0$ otherwise.
We can simulate product using contraction by including additional copy tensors.
In detail, for each $z \in X \cap Y$ make two fresh variables $z'$ and $z''$.
Replace each $z$ in $f$ with $z'$ to produce $f'$, and replace each $z$ in $g$ with $z''$ to produce $g'$.
Then one can check that $f \cdot g = f' \contract g' \contract \bigcontract_{z \in X \cap Y} \blacksquare_{\{z, z', z''\}}$.

When a product is immediately followed by the projection of shared variables (\ie, we are computing $\proj_Z f \cdot g$ for some $Z \subseteq X \cap Y$), we can optimize this procedure.
In particular, we skip creating copy tensors for the variables in $Z$ and instead eliminate them directly as we perform $f' \contract g'$.
In this case, we do not ever fully compute $f \cdot g$, so the maximum number of variables needed in each intermediate tensor may be lower than the width of the project-join tree.
In the context of tensor networks and contraction trees, the maximum number of variables needed after accounting for this optimization is the \emph{max-rank} of the contraction tree \cite{kourtis2019fast,dudek2019efficient}.
The max-rank is often lower than the width of the corresponding project-joint tree.
On the other hand, the intermediate terms in the computation of $f \cdot g$ with contractions may have more variables than either $f$, $g$, or $f \cdot g$.
Thus the number of variables in each intermediate tensor may be higher than the width of the project-join tree (by at most a factor of 1.5).

\section{Empirical Evaluation}
\label{sec_experiments}

We are interested in the following experimental research questions, where we aim to answer each research question with an experiment.
\begin{itemize}
    \item[(RQ1)] In the planning phase, how do constraint-satisfaction heuristics compare to tree-decomposition solvers?
    \item[(RQ2)] In the execution phase, how do ADDs compare to tensors as the underlying data structure?
    \item[(RQ3)] Are project-join-tree-based weighted model counters competitive with state-of-the-art tools?
\end{itemize}

To answer RQ1, we build two implementations of the planning phase: \Htb{} (for Heuristic Tree Builder, based on \cite{dudek2020addmc}) and \Lg{} (for Line Graph, based on \cite{dudek2019efficient}).
\Htb{} implements Algorithm \ref{alg_csp_jt} and so is representative of the constraint-satisfaction approach.
\Htb{} contains implementations of four clustering heuristics (\Be-\ListH, \Be-\TreeH, \Bm-\ListH, and \Bm-\TreeH) and nine cluster-variable-order heuristics (\Random, \Mcs, \Invmcs, \Lexp, \Invlexp, \Lexm, \Invlexm, \Minfill, and \Invminfill).
\Lg{} implements Algorithm \ref{alg_td_to_join} and so is representative of the tree-decomposition approach.
In order to find tree decompositions, \Lg{} leverages three state-of-the-art heuristic tree-decomposition solvers: \Flowcutter{} \cite{strasser2017computing}, \Htd{} \cite{abseher2017htd}, and \Tamaki{} \cite{tamaki2019positive}.
These solvers are all \emph{anytime}, meaning that \Lg{} never halts but continues to produce better and better project-join trees when given additional time.
On the other hand, \Htb{} produces a single project-join tree.
We compare these implementations on the planning phase in Section \ref{sec_experiments_planning}.

To answer RQ2, we build two implementations of the execution phase: \Dmc{} (for Diagram Model Counter, based on \cite{dudek2020addmc}) and \Tensor{} (based on \cite{dudek2019efficient}).
\Dmc{} uses ADDs as the underlying data structure with \cudd{} \cite{somenzi2015cudd}.
\Tensor{} uses tensors as the underlying data structure with \Numpy{} \cite{numpy}.
We compare these implementations on the execution phase in Section \ref{sec_experiments_execution}.
Since \Lg{} is an anytime tool, each execution tool must additionally determine the best time to terminate \Lg{} and begin performing the valuation.
We explore options for this in Section \ref{sec_experiments_execution}.

To answer RQ3, we combine each implementation of the planning phase and each implementation of the execution phase to produce model counters that use project-join trees.
We then compare these model counters with the state-of-the-art tools \cachet{} \cite{sang2004combining}, \ctd{} \cite{darwiche2004new}, \df{} \cite{lagniez2017improved}, and \minictd{} \cite{oztok2015top} in Section \ref{sec_experiments_wmc}.

We use a set of \benchmarkCountAltogether{} literal-weighted model counting benchmarks from \cite{dudek2020addmc}.
These benchmarks were gathered from two sources.
First, the \classBayes{} class%
\footnote{\urlBenchmarksBayes}
consists of \benchmarkCountBayes{} CNF benchmarks%
\footnote{excluding 11 benchmarks double-counted by \cite{dudek2020addmc}}
that encode Bayesian inference problems \cite{sang2005performing}.
All literal weights in this class are between 0 and 1. 
Second, the \classOther{} class%
\footnote{\urlBenchmarksOther}
consists of \benchmarkCountOther{} CNF benchmarks%
\footnote{including 73 benchmarks missed by \cite{dudek2020addmc}}
that are divided into eight families: \famBmc, \famCircuit, \famConfig, \famHandmade, \famPlanning, \famQif, \famRandom, and \famSchedule{} \cite{clarke2001bounded,sinz2003formal,palacios2009compiling,klebanov2013sat}.
All \classOther{} benchmarks are originally unweighted.
As we focus in this work on weighted model counting, we generate weights for these benchmarks.
Each variable $x$ is randomly assigned literal weights: either $W_x(\set{x}) = 0.5$ and $W_x(\emptyset) = 1.5$, or $W_x(\set{x}) = 1.5$ and $W_x(\emptyset) = 0.5$.
Generating weights in this particular fashion results in a reasonably low amount of floating-point underflow and overflow for all model counters.

We ran all experiments on single CPU cores of a Linux cluster with Xeon E5-2650v2 processors (2.60-GHz) and 30 GB of memory.
All code, benchmarks, and experimental data are available in a public repository (\url{https://github.com/vardigroup/DPMC}).


\subsection{Experiment 1: Comparing Project-Join Planners}
\label{sec_experiments_planning}

\begin{figure}[t]
	\centering
	\input{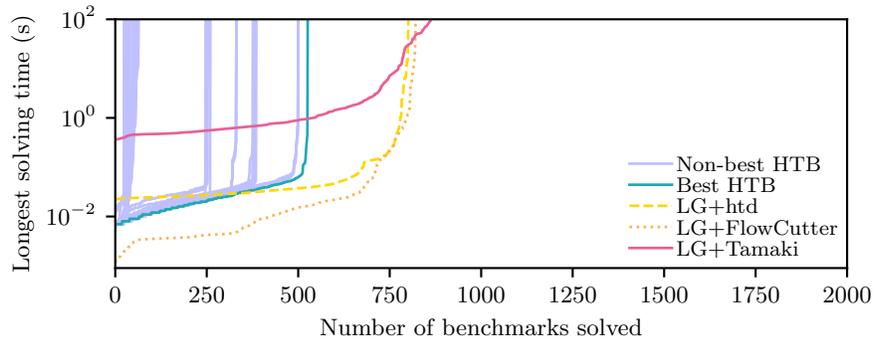}
    \vspace*{-1cm}
	\caption{\label{fig:planning} A cactus plot of the performance of various planners.
	A planner ``solves'' a benchmark when it finds a project-join tree of width 30 or lower.}
\end{figure}

We first compare constraint-satisfaction heuristics (\Htb) and tree-decomposition tools (\Lg) at building project-join trees of CNF formulas.
To do this, we ran all 36 configurations of \Htb{} (combining four clustering heuristics with nine cluster-variable-order heuristics) and all three configurations of \Lg{} (choosing a tree-decomposition solver) once on each benchmark with a 100-second timeout.
In Figure \ref{fig:planning}, we compare how long it takes various methods to find a high-quality (meaning width at most 30) project-join tree of each benchmark.
We chose 30 for Figure \ref{fig:planning} since \cite{dudek2019efficient} observed that tensor-based approaches were unable to handle trees whose widths are above 30, but Figure \ref{fig:planning} is qualitatively similar for other choices of widths.
We observe that \Lg{} is generally able to find project-join trees of lower widths than those \Htb{} is able to find.
We therefore conclude that tree-decomposition solvers outperform constraint-satisfaction heuristics in this case.
We observe that \Be-\TreeH{} as the clustering heuristic and \Invlexp{} as the cluster-variable-order heuristic make up the best-performing heuristic configuration from \Htb.
This was previously observed to be the second-best heuristic configuration for weighted model counting in \cite{dudek2020addmc}.
We therefore choose \Be-\TreeH{} with \Invlexp{} as the representative heuristic configuration for \Htb{} in the remaining experiments.
For \Lg{}, we choose \Flowcutter{} as the representative tree-decomposition tool in the remaining experiments.


\subsection{Experiment 2: Comparing Execution Environments}
\label{sec_experiments_execution}

\begin{figure}[t]
	\centering
	\input{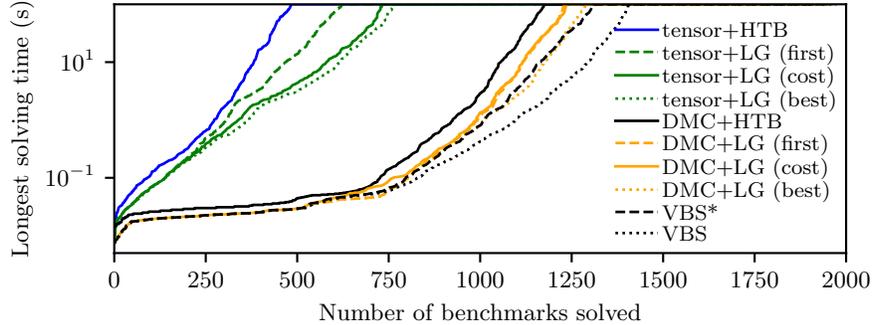}
    \vspace*{-1cm}
	\caption{
	A cactus plot of the performance of various planners and executors for weighted model counting.
    Different strategies for stopping \Lg{} are considered.
    ``(first)'' indicates that \Lg{} was stopped after it produced the first project-join tree.
    ``(cost)'' indicates that the executor attempted to predict the cost of computing each project-join tree.
    ``(best)'' indicates a simulated case where the executor has perfect information on all project-join trees generated by \Lg{} and valuates the tree with the shortest total time.
    \tool{VBS*} is the virtual best solver of \Dmc{}+\Htb{} and \Dmc{}+\Lg{} (cost).
	\tool{VBS} is the virtual best solver of \Dmc{}+\Htb{}, \Dmc{}+\Lg{} (cost), \Tensor{}+\Htb{}, and \Tensor{}+\Lg{} (cost).}
	\label{fig:execution}
\end{figure}

Next, we compare ADDs (\Dmc) and tensors (\Tensor) as a data structure for valuating project-join trees.
To do this, we ran both \Dmc{} and \Tensor{} on all project-join trees generated by \Htb{} and \Lg{} (with their representative configurations) in Experiment 1, each with a 100-second timeout.
The total times recorded include both the planning stage and the execution stage.

Since \Lg{} is an anytime tool, it may have produced more than one project-join tree of each benchmark in Experiment 1.
We follow \cite{dudek2019efficient} by allowing \Tensor{} and \Dmc{} to stop \Lg{} at a time proportional to the estimated cost to valuate the best-seen project-join tree.
The constant of proportionality is chosen to minimize the PAR-2 score (\ie, the sum of the running times of all completed benchmarks plus twice the timeout for every uncompleted benchmark) of each executor.
\Tensor{} and \Dmc{} use different methods for estimating cost.
Tensors are a dense data structure, so the number of floating-point operations to valuate a project-join tree can be computed exactly as in \cite{dudek2019efficient}.
We use this as the cost estimator for \Tensor{}.
ADDs are a sparse data structure, and estimating the amount of sparsity is difficult.
It is thus hard to find a good cost estimator for \Dmc{}.
As a first step, we use $2^w$ as an estimate of the cost for \Dmc{} to valuate a project-join tree of width $w$.

We present results from this experiment in Figure \ref{fig:execution}.
We observe that the benefit of \Lg{} over \Htb{} seen in Experiment 1 is maintained once the full weighted model count is computed.
We also observe that \Dmc{} is able to solve significantly more benchmarks than \Tensor{}, even when using identical project-join trees.
We attribute this difference to the sparsity of ADDs over tensors.
Nevertheless, we observe that \Tensor{} still outperforms \Dmc{} on some benchmarks; compare \tool{VBS*} (which excludes \Tensor{}) with \tool{VBS} (which includes \Tensor{}).

Moreover, we observe significant differences based on the strategy used to stop \Lg{}.
The executor \Tensor{} performs significantly better when cost estimation is used than when only the first project-join tree of \Lg{} is used.
In fact, the performance of \Tensor{} is almost as good as the hypothetical performance if \Tensor{} is able to predict the planning and valuation times of all trees produced by \Lg{}.
On the other hand, \Dmc{} is not significantly improved by cost estimation.
It would be interesting in the future to find better cost estimators for \Dmc{}.


\subsection{Experiment 3: Comparing Exact Weighted Model Counters}
\label{sec_experiments_wmc}

\begin{figure}[t]
	\centering
	\input{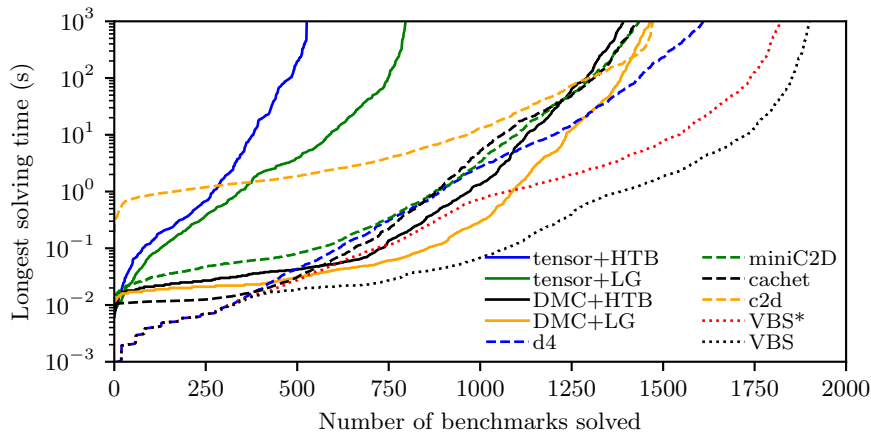}
    \vspace*{-1cm}
	\caption{\label{fig:comparison} A cactus plot of the performance of four project-join-tree-based model counters, two state-of-the-art model counters, and two virtual best solvers: \tool{VBS*} (without project-join-tree-based counters) and \tool{VBS} (with project-join-tree-based counters).}
\end{figure}

Finally, we compare project-join-tree-based model counters with state-of-the-art tools for weighted model counting.
We construct four project-join-tree-based model counters by combining \Htb{} and \Lg{} (using the representative configurations from Experiment 1) with \Dmc{} and \Tensor{} (using the cost estimators for \Lg{} from Experiment 2).
Note that \Dmc{}+\Htb{} is equivalent to \tool{ADDMC} \cite{dudek2020addmc}, and \Tensor{}+\Lg{} is equivalent to \tool{TensorOrder} \cite{dudek2019efficient}.
We compare against the state-of-the-art model counters \cachet{} \cite{sang2004combining}, \ctd{} \cite{darwiche2004new}, \df{} \cite{lagniez2017improved}, and \minictd{} \cite{oztok2015top}.
We ran each benchmark once with each model counter with a 1000-second timeout and recorded the total time taken.
For the project-join-tree-based model counters, time taken includes both the planning stage and the execution stage.

We present results from this experiment in Figure \ref{fig:comparison}.
For each benchmark, the solving time of \tool{VBS*} is the shortest solving time among all pre-existing model counters (\cachet, \ctd, \df, and \minictd).
Similarly, the time of \tool{VBS} is the shortest time among all model counters, including those based on project-join trees.
We observe that \tool{VBS} performs significantly better than \tool{VBS*}.
In fact, \Dmc{}+\Lg{} is the fastest model counter on 471
of \benchmarkCountAltogether{}
benchmarks.
Thus project-join-tree-based tools are valuable for portfolios of weighted model counters.

\section{Discussion}
\label{sec_discussion}

In this work, we introduced the concept of project-join trees for weighted model counting.
These trees are at the center of a dynamic-programming framework that unifies and generalizes several model counting algorithms, including those based on ADDs \cite{dudek2020addmc}, tensors \cite{dudek2019efficient}, and database management systems \cite{fichte2020exploiting}.
This framework performs model counting in two phases.
First, the planning phase produces a project-join tree from a CNF formula.
Second, the execution phase uses the project-join tree to guide the dynamic-programming computation of the model count of the formula \wrt{} a literal-weight function.
The current implementation of our dynamic-programming model-counting framework \Dpmc{} includes two planners (\Htb{} and \Lg) and two executors (\Dmc{} and \Tensor{}).

For the planning phase, we implemented \Htb{} based on constraint-satisfaction heuristics
\cite{tarjan1984simple,koster2001treewidth,dechter2003constraint,dechter1999bucket,bouquet1999gestion}
and \Lg{} based on tree-decomposition tools \cite{strasser2017computing,tamaki2019positive,abseher2017htd}.
Our empirical work indicates that tree-decomposition tools tend to produce project-join trees of lower widths in shorter times.
This is a significant finding with applications beyond model counting, \eg, in Boolean functional synthesis \cite{tabajara2017factored}.

For the execution phase, we implemented \Dmc{} based on ADDs \cite{dudek2020addmc,somenzi2015cudd} and \Tensor{} based on tensors \cite{dudek2019efficient,numpy}.
Empirically, we observed that (sparse) ADDs outperform (dense) tensors on single CPU cores.
Whether this holds for richer architectures as well is a subject for future work.
We will also consider adding to our framework an executor based on databases (\eg, \cite{fichte2020exploiting}).

We showed that our dynamic-programming model-counting framework \Dpmc{} is competitive with state-of-the-art tools (\cachet{} \cite{sang2004combining}, \ctd{} \cite{darwiche2004new}, \df{} \cite{lagniez2017improved}, and \minictd{} \cite{oztok2015top}).
Although no single model counter dominates,
\Dpmc{} considerably improves the virtual best solver and thus is valuable as part of the portfolio.

In this work, we did not consider preprocessing of benchmarks.
For example, \cite{dudek2019efficient} found that preprocessing (called \pkg{FT}, based on a technique to reduce variable occurrences using tree decompositions of the incidence graph \cite{samer2010constraint}) significantly improved tensor-network-based approaches for weighted model counting.
Moreover, \cite{fichte2019improved} and \cite{dudek2020parallel} observed that the \tool{pmc} preprocessor \cite{lagniez2014preprocessing} notably improved the running time of some dynamic-programming-based model counters.
We expect these techniques to also improve \tool{DPMC}.

A promising future research direction is multicore programming.
Our planning tool \Lg{} can be improved to run back-end tree-decomposition solvers in parallel, as in \cite{dudek2020parallel}.
We can also make the execution tool \Dmc{} support multicore ADD packages (\eg, \sylvan{} \cite{van2015sylvan}).
Our other executor, \Tensor{}, is built on top of \Numpy{} \cite{numpy} and should be readily parallelizable (\eg, using techniques from \cite{dudek2020parallel}).
We can then compare \Dpmc{} to parallel solvers (\eg, \cite{dal2018parallel,burchard2015laissez}).

Finally, decision diagrams have been widely used in artificial intelligence in the context of \emph{knowledge compilation}, where formulas are compiled into a tractable form in an early phase to support efficient query processing \cite{koriche2013knowledge,lagniez2017improved,darwiche2004new,oztok2015top}.
Our work opens up an investigation into the combination of knowledge compilation and dynamic programming.
The focus here is on processing a single model-counting query.
Exploring how dynamic programming can also be leveraged to handle several queries is another promising research direction.


\bibliographystyle{templates/splncs04}
\bibliography{DPMC}


\clearpage
\appendix

\section{Constraint-Satisfaction Heuristics for Project-Join Tree Planning}
\label{sec_csp_heuristics}


\subsection{Heuristics for $\clusterVarOrder$}

In Algorithm \ref{alg_csp_jt}, the function $\clusterVarOrder$ returns a variable order what will be used to rank the clauses of $\phi$.
We consider nine heuristics for variable ordering: \Random, \Mcs, \Lexp, \Lexm, \Minfill, \Invmcs, \Invlexp, \Invlexm, and \Invminfill.

One simple heuristic for $\clusterVarOrder$ is to randomly order the variables, \ie, for a formula over some set $X$ of variables, sample an injection $X \to \set{1, 2, \ldots, |X|}$ uniformly at random.
We call this the \Random{} heuristic.
\Random{} is a baseline to compare other variable-order heuristics.

For the remaining heuristics, we use {Gaifman graphs} of CNF formulas. 
Recall that the Gaifman graph $\gaifman(\phi)$ of a CNF formula $\phi$ has a vertex for each variable of $\phi$.
Also, two vertices of $\gaifman(\phi)$ are connected by an edge if and only if the corresponding variables appear together in some clause of $\phi$.
We say that two variables of $\phi$ are \emph{adjacent} if the corresponding two vertices of $\gaifman(\phi)$ are neighbors.

A well-known heuristic for $\clusterVarOrder$ is \textdef{maximum-cardinality search} \cite{tarjan1984simple}.
At each step of the heuristic, the next variable chosen is the variable adjacent to the greatest number of previously chosen variables.
We call this the \Mcs{} heuristic for variable ordering.

Another heuristic is \textdef{lexicographic search for perfect orders} \cite{koster2001treewidth}.
Every vertex $v$ of $\gaifman(\phi)$ is assigned an initially empty set of vertices, called the \textdef{label} of $v$.
At each step of the heuristic, the next variable chosen is the variable $x$ whose label is lexicographically smallest among the unchosen variables.
Then $x$ is added to the labels of its neighbors in $\gaifman(\phi)$.
We call this the \Lexp{} heuristic for variable ordering.

A similar heuristic is \textdef{lexicographic search for minimal orders} \cite{koster2001treewidth}.
As before, each vertex of $\gaifman(\phi)$ is assigned an initially empty label.
At each step of the heuristic, the next variable chosen is again the variable $x$ whose label is lexicographically smallest.
Then $x$ is added to the label of every variable $y$ \st{} there is a path $x, z_1, z_2, \ldots, z_k, y$ in $\gaifman(\phi)$ where every $z_i$ is unchosen and the label of $z_i$ is lexicographically smaller than the label of $y$.
We call this the \Lexm{} heuristic for variable ordering.

A different heuristic is \textdef{minimal fill-in} \cite{dechter2003constraint}.
Whenever a variable $v$ is chosen, we add \textdef{fill-in} edges to connect all of $v$'s neighbors in the Gaifman graph.
At each step of the heuristic, the next variable chosen is the variable that minimizes the number of fill-in edges.
We call this the \Minfill{} heuristic for variable ordering.

Additionally, the variable orders produced by \Mcs{}, \Lexp{}, \Lexm, and \Minfill{} can be inverted.
We call these heuristics \Invmcs, \Invlexp, \Invlexm, and \Invminfill.


\subsection{Heuristics for $\clauseRank$}

In Algorithm \ref{alg_csp_jt}, given a cluster variable order $\rho$, we partition the clauses of $\phi$ by calling the function $\clauseRank$.
We consider two possible heuristics for $\clauseRank$ that satisfy the conditions of Theorem \ref{thm_csp_jt}: \Be{} and \Bm{}.

One heuristic assigns the rank of each clause $c \in \phi$ to be the smallest $\rho$-rank of the variables that appear in $c$, \ie, $\clauseRank(c, \rho) = \min_{x \in \vars(c)} \rho(x)$.
This heuristic corresponds to \textdef{bucket elimination} \cite{dechter1999bucket}, so we call it the \Be{} heuristic.
Using \Be{} for $\clauseRank$ in Algorithm \ref{alg_csp_jt}, notice that every CNF clause $c$ containing a variable $x \in X$ can only appear in a set $\Gamma_i$ if $i \le \rho(x)$.
It follows that $x$ has always been projected from all clauses by the end of iteration $\rho(x)$ in the second loop.

A different heuristic assigns the rank of each clause to be the largest $\rho$-rank of the variables that appear in the clause.
That is, $\clauseRank(c, \rho) = \max_{x \in \vars(c)} \rho(x)$.
This heuristic corresponds to \textdef{Bouquet's Method} \cite{bouquet1999gestion}, so we call it the \Bm{} heuristic.
Unlike the \Be{} case, we can make no guarantee about when each variable is projected in Algorithm \ref{alg_csp_jt} using \Bm{}.


\subsection{Heuristics for $\chosenCluster$}

In Algorithm \ref{alg_csp_jt}, the function $\chosenCluster$ determines the parent of the current internal node.
We consider two possible heuristics to use for $\chosenCluster$ that satisfy the conditions of Theorem \ref{thm_csp_jt}: \heuristic{List} and \heuristic{Tree} \cite{dudek2020addmc}.

One option is for $\chosenCluster$ to place the internal node $n_i$ in the nearest cluster that satisfies the conditions of Theorem \ref{thm_csp_jt}, namely the next cluster to be processed.
That is, $\chosenCluster(n_i) = i + 1$.
We call this the \heuristic{List} heuristic.
Notice that project-join trees are left-deep with \ListH{}.

Another option is for $\chosenCluster$ to place $n_i$ in the furthest cluster that satisfies the conditions of Theorem \ref{thm_csp_jt}.
That is, $\chosenCluster(n_i)$ returns the smallest $j > i$ \st{} $X_j \cap \vars(n_i) \ne \emptyset$ (or returns $m$, if $\vars(n_i) = \emptyset)$.
We call this the \heuristic{Tree} heuristic.
Project-join trees with the \TreeH{} heuristic are more balanced than those with the \ListH{} heuristic.

\section{Examples}


Figure \ref{fig_add} illustrates an algebraic decision diagram (ADD).

\begin{figure}
  \centering
  \includegraphics[
    trim={2in .5in .1in .5in} 
  ]
  {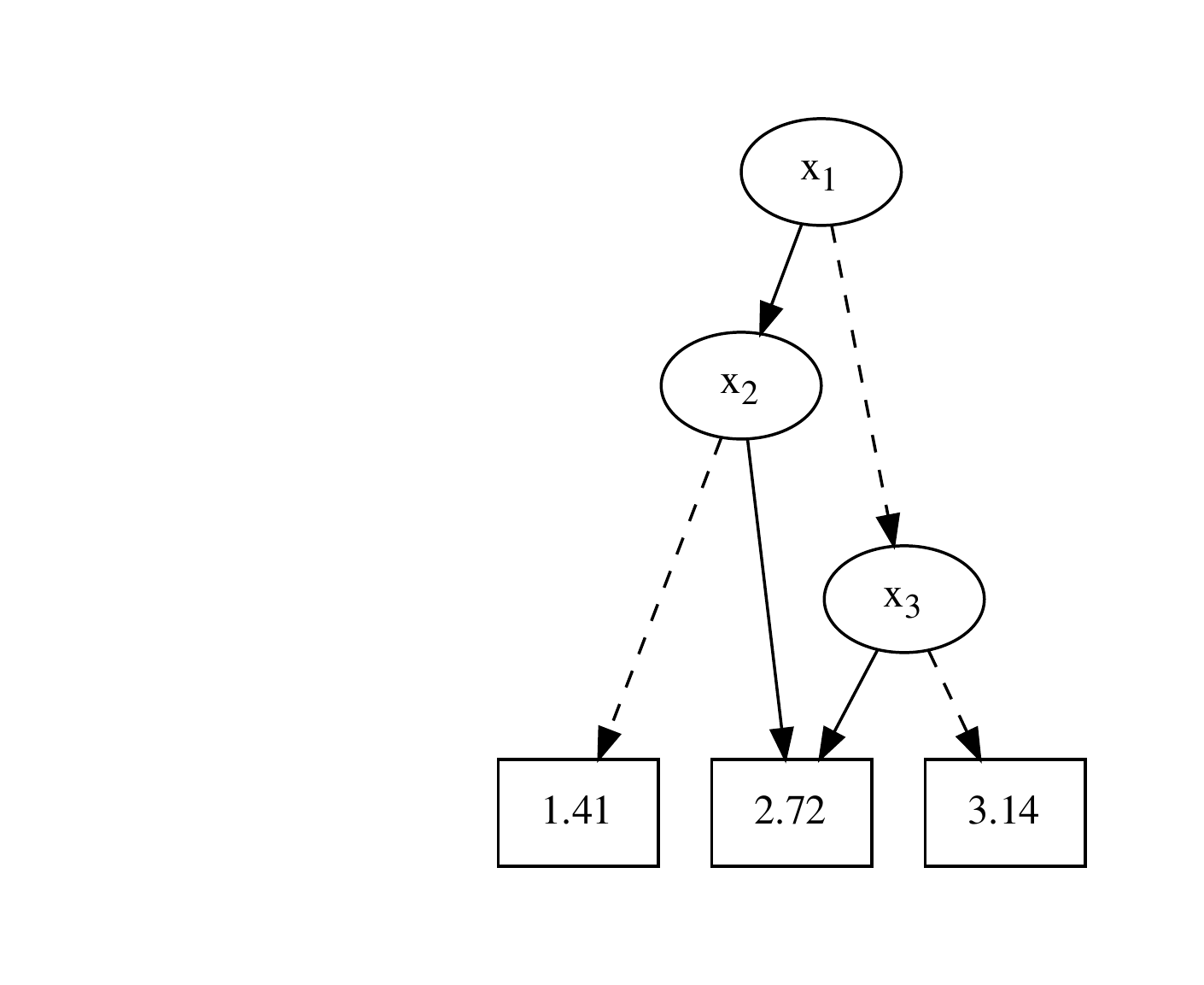}
  \caption{
    The directed graph $G$ of an ADD with variable set $X = \set{x_1, x_2, x_3}$, carrier set $S = \R$, and diagram variable order $\pi(x_i) = i$ for $i = 1, 2, 3$.
    If an edge from an oval node is solid (respectively dashed), then the corresponding Boolean variable is assigned 1 (respectively 0).
  }
\label{fig_add}
\end{figure}

\section{Proofs}
\label{sec_proofs}


\subsection{Proof of Theorem \ref{thm_early_proj}}

\begin{proof}
    For every $\tau \in 2^{(X \cup Y) \setminus \set{x}}$, we have:
    \begin{align*}
      \pars{\proj_x (A \mult B)} &(\tau)
       = (A \mult B)(\tau) + (A \mult B)(\tau \cup \set{x})
        \tag{Definition \ref{def_proj}} \\
      & = A(\tau \cap X) \mult B(\tau \cap Y)+ A((\tau \cup \set{x}) \cap X) \mult B((\tau \cup \set{x}) \cap Y)
        \tag{Definition \ref{def_mult}} \\
      & = A(\tau \cap X) \mult B(\tau \cap Y)+ A((\tau \cup \set{x}) \cap X) \mult B(\tau \cap Y)
        \tag{as $x \notin Y$} \\
      & = A(\tau \cap X) \mult B(\tau \cap Y)+ A(\tau \cap X \cup \set{x}) \mult B(\tau \cap Y)
        \tag{as $x \in X$} \\
      & = (A(\tau \cap X) + A(\tau \cap X \cup \set{x})) \mult B(\tau \cap Y) \\
      & = \pars{\proj_x A}(\tau \cap X) \mult B(\tau \cap Y)
        \tag{Definition \ref{def_proj}} \\
      & = \pars{\proj_x A}(\tau \cap (X \setminus \set{x})) \mult B(\tau \cap Y)
        \tag{as $x \notin \tau$} \\
      & = \pars{\pars{\proj_x A} \mult B)}(\tau)
        \tag{Definition \ref{def_mult}}
    \end{align*}
\qed
\end{proof}


\subsection{Proof of Theorem \ref{thm_valuation_wmc}}

In this section, for pseudo-Boolean functions $f : 2^X \to \R$, we define $\vars(f) \equiv X$.
Notice that a clause $c$ in a CNF formula can be interpreted as a Boolean function $c : 2^{\vars(c)} \to \B$.
Similarly, a set $\phi$ of clauses can be interpreted as the Boolean function $\phi = \prod_{c \in \phi} c$.

Let $X$ be a set of variables and $W = \prod_{x \in X} W_x$ be a literal-weight function.
Given a set $Y \subseteq X$, we define $W_{Y} \equiv \prod_{x \in Y} W_x$.
Notice $W_Y \mult W_Z = W_{Y \cup Z}$ for all sets $Y, Z \subseteq X$.

Let $\phi$ be a CNF formula over a set $X$ of variables, $(T, r, \gamma, \sigma)$ be a project-join tree of $\phi$, and $n \in \V T$.
Denote by $S(n)$ the subtree rooted at $n$.
We define the set $\Phi(n)$ of clauses that correspond to the leaves of $S(n)$:
\begin{align*}
    \Phi(n) \equiv
    \begin{cases}
        \set{\gamma(n)} & \text{if } n \in \Lv T \\
        \bigcup_{o \in \C(n)} \Phi(o) & \text{otherwise}
    \end{cases}
\end{align*}
We also define the set $P(n)$ of all variables to project in the subtree $S(n)$:
\begin{align*}
    P(n) \equiv
    \begin{cases}
       \emptyset & \text{if } n \in \Lv{T} \\
        \pi(n) \cup \bigcup_{o \in \C(n)} P(o) & \text{otherwise}
    \end{cases}
\end{align*}

\begin{lemma}
\label{lemma_disjoint_P}
    In a project-join tree $(T, r, \gamma, \pi)$, let $n$ be an internal node with children $o \ne q$.
    Then $P(o) \cap \vars \pars{\Phi(q) \mult W_{P(q)}} = \emptyset$.
\end{lemma}
\begin{proof}
    Let variable $x \in P(o)$.
    Notice that $x \in \pi(s)$ for some internal node $s$ that is a descendant of $o$.
    Assume there is an arbitrary clause $c \in \phi$ \st{} $x$ appears in $c$.
    By the last property in Definition \ref{def_jointree}, the corresponding leaf $\gamma^{-1}(c)$ is a descendant of $s$ and thus a descendant of $o$.
    So $x$ appears in no descendant leaf of $q$ (as $q$ is a sibling of $o$ in the tree $T$).
    Thus $x \notin \vars(\Phi(q))$.

    Now, note that $P(q) \subseteq \vars(\Phi(q))$, again by Definition \ref{def_jointree}.
    So $x \notin P(q)$.
    Therefore $x \notin \vars(\Phi(q) \mult W_{P(q)}) = \vars(\Phi(q)) \cup \vars(W_{P(q)}) = \vars(\Phi(q)) \cup P(q)$.
    Since $x \in P(o)$ is arbitrary, we have $P(o) \cap \vars \pars{\Phi(q) \mult W_{P(q)}} = \emptyset$.
\qed
\end{proof}

\begin{lemma}
\label{lemma_late_proj}
    In a project-join tree $(T, r, \gamma, \pi)$, let $n$ be an internal node with children $o \ne q$.
    Then:
    $$\proj_{P(o)} \pars{ \Phi(o) \mult W_{P(o)} } \mult \proj_{P(q)} \pars{ \Phi(q) \mult W_{P(q)} } = \proj_{P(o) \cup P(q)} \pars{ \Phi(o) \Phi(q) \mult W_{P(o) \cup P(q)} }$$
\end{lemma}
\begin{proof}
    We have:
    \begin{align*}
        \proj_{P(o)} \pars{ \Phi(o) \mult W_{P(o)} } \mult \proj_{P(q)} \pars{ \Phi(q) \mult W_{P(q)} }
        & = \proj_{P(o)} \pars{ \Phi(o) \mult W_{P(o)} \mult \proj_{P(q)} \pars{ \Phi(q) \mult W_{P(q)} } } \tag{undoing early projection of $P(o)$, observing Lemma \ref{lemma_disjoint_P}} \\
        & = \proj_{P(q)} \proj_{P(o)} \pars{ \Phi(o) \mult W_{P(o)} \mult \Phi(q) \mult W_{P(q)} } \tag{undoing early projection of $P(q)$, observing Lemma \ref{lemma_disjoint_P}} \\
        & = \proj_{P(o)} \proj_{P(q)} \pars{ \Phi(o) \Phi(q) \mult W_{P(o)} W_{P(q)} }
        \\
        & = \proj_{P(o) \cup P(q)} \pars{ \Phi(o) \Phi(q) \mult W_{P(o) \cup P(q)} }
    \end{align*}
\qed
\end{proof}

\pagebreak

\begin{lemma}
\label{lemma_valuation}
    Let $\phi$ be a CNF formula over a set $X$ of variables, $W$ be a literal-weight function over $X$, and $(T, r, \gamma, \pi)$ be a project-join tree of $\phi$.
    Then for every node $n$ in $T$:
    $$f^W_n = \proj_{P(n)} \pars{ \Phi(n) \mult W_{P(n)} }$$
\end{lemma}
\begin{proof}
    Notice that both pseudo-Boolean functions have the same variables in their domains:
    \begin{align*}
        \vars \pars{f^W_n}
        & = \vars(n) \\
        & = \vars(\Phi(n)) \setminus P(n) \\
        & = \vars \pars{ \proj_{P(n)} \pars{ \Phi(n) \mult W_{P(n)} } }
    \end{align*}

    We employ structural induction on $n \in \V T$.
    In the base case, $n$ is a leaf.
    So $P(n) = \emptyset$ and $\Phi(n) = \set{\gamma(n)}$.
    We have:
    \begin{align*}
        f^W_n
        & = \gamma(n) \tag{Equation \eqref{eq_valuation}} \\
        & = \prod_{c \in \Phi(n)} c \tag{singleton set} \\
        & = \Phi(n)
        \\
        & = \Phi(n) \mult \prod_{x \in \emptyset} W_x \tag{an empty product is equal to the number $1$} \\
        & = \Phi(n) \mult W_\emptyset \\
        & = \Phi(n) \mult W_{P(n)} \\
        & = \proj_{\emptyset} \pars{ \Phi(n) \mult W_{P(n)} } \tag{convention on projection} \\
        & = \proj_{P(n)} \pars{ \Phi(n) \mult W_{P(n)} }
    \end{align*}

    In the inductive case, $n$ is an internal node of $T$.
    Our induction hypothesis is that
    $$f^W_o = \proj_{P(o)} \pars{ \Phi(o) \mult W_{P(o)} }$$
    for every child node $o$ of $n$.
    Then we have:
    \begin{align*}
        f^W_n
        & = \proj_{\pi(n)} \pars{\prod_{o \in \C(n)} f^W_o \mult \prod_{x \in \pi(n)} W_x} \tag{Equation \eqref{eq_valuation}} \\
        & = \proj_{\pi(n)} \pars{\prod_{o \in \C(n)} f^W_o \mult W_{\pi(n)}} \\
        & = \proj_{\pi(n)} \pars{\prod_{o \in \C(n)} \pars{\proj_{P(o)} \pars{ \Phi(o) \mult W_{P(o)} } } \mult W_{\pi(n)}} \tag{induction hypothesis} \\
        & = \proj_{\pi(n)} \pars{ \proj_{\bigcup_{s \in \C(n)} P(s)} \pars{ \prod_{o \in \C(n)} \Phi(o) \mult W_{\bigcup_{t \in \C(n)} P(t)} } \mult W_{\pi(n)} } \tag{applying Lemma \ref{lemma_late_proj} multiple times} \\
        & = \proj_{\pi(n)} \pars{ \proj_{\bigcup_{s \in \C(n)} P(s)} \pars{ \prod_{o \in \C(n)} \Phi(o) \mult W_{\bigcup_{t \in \C(n)} P(t)} \mult W_{\pi(n)}} } \tag{undoing early projection, observing $\pi(n) \cap \bigcup_{s \in \C(n)} P(s) = \emptyset$} \\
        & = \proj_{\pi(n) \cup \bigcup_{s \in \C(n)} P(s)} \pars{ \prod_{o \in \C(n)} \Phi(o) \mult W_{\bigcup_{t \in \C(n)} P(t) \cup \pi(n)} } \\
        & = \proj_{P(n)} \pars{ \prod_{o \in \C(n)} \Phi(o) \mult W_{P(n)} } \tag{definition of $P(n)$} \\
        & = \proj_{P(n)} \pars{ \Phi(n) \mult W_{P(n)} } \tag{as $\Phi(n) = \bigcup_{o \in \C(n)} \Phi(o)$ is a set of clauses}
    \end{align*}
\qed
\end{proof}

Now, we can prove Theorem \ref{thm_valuation_wmc}.
\begin{proof}
    As $r$ is the root of the project-join tree, $P(r) = X$ and $\Phi(r) = \phi$.
    Then:
    \begin{align*}
        f^W_r(\emptyset)
        & = \pars{ \proj_{P(r)} \pars{ \Phi(r) \mult W_{P(r)} } } (\emptyset) \tag{Lemma \ref{lemma_valuation}} \\
        & = \pars{ \proj_{X} \pars{ \phi \mult W_{X} } } (\emptyset) \\
        & = \pars{ \proj_{X} \pars{ \phi \mult W } } (\emptyset) \\
        & = W(\phi)
    \end{align*}
\qed
\end{proof}


\subsection{Proof of Theorem \ref{thm_csp_jt}}

In this section, we assume the antecedents of Theorem \ref{thm_csp_jt} regarding the functions $\clusterVarOrder$, $\clauseRank$, and $\chosenCluster$.
Notice that for each $i = 1, 2, \ldots, m$ in Algorithm \ref{alg_csp_jt}, we have the following:
\begin{itemize}
    \item $\Gamma_i$ is a set of clauses
    \item $\kappa_i$ is a set of nodes that includes leaves $l$ \st{} $\gamma(l) \in \Gamma_i$
    \item $n_i$ is an internal node
    \item $n_i$'s children include the leaves in $\kappa_i$
    \item $\pi(n_i) = X_i$
\end{itemize}

We show that the first property in Definition \ref{def_jointree} holds:
\begin{lemma}[Property \ref{prop1}]
\label{lemma_prop1}
    The set $\set{\pi(n) : n \in \V T \setminus \Lv T}$ is a partition of $X$.
\end{lemma}
\begin{proof}
    For each $i = 1, 2, \ldots, m$, Algorithm \ref{alg_csp_jt} constructs an internal nodes $n_i$ with $\pi(n_i) = X_i$.
    Recall that $\set{X_i}_{i = 1}^m$ is a partition of $X$.
    Then $\set{\pi(n_i)}_{i = 1}^m$ is the same partition of $X$.
\qed
\end{proof}

We show that the second property in Definition \ref{def_jointree} holds through the following lemmas.

\begin{lemma}
\label{lemma_unprojected}
    Let $1 \le p < q \le m$.
    Assume some $x \in \vars(\Gamma_p) \cap X_q$.
    Then $x \in \vars(n_p)$.
\end{lemma}
\begin{proof}
    Notice $x \in X_q = \pi(n_q)$.
    Then $x$ is projected in $n_q$.
    Since $p < q$, we know $x$ is projected in neither $n_p$ nor a descendants of $n_p$.
    Since $x \in \vars(\Gamma_p)$, we know $x$ remains in $n_p$.
\qed
\end{proof}

\begin{lemma}
\label{lemma_internal_descedant}
    Let $1 \le p_0 < q \le m$.
    Assume $\vars(\Gamma_{p_0}) \cap X_q \ne \emptyset$.
    Then the internal node $n_{p_0}$ is a descendant of the node $n_q$.
\end{lemma}
\begin{proof}
    Let $n_{p_1}, n_{p_2}, \ldots$ be the parent, grandparent,\ldots{} of $n_{p_0}$.
    By way of contradiction, assume every $p_i \ne q$.
    Let $x$ be a variable in $\vars(\Gamma_{p_0}) \cap X_q \ne \emptyset$.
    By Lemma \ref{lemma_unprojected}, we know $x \in \vars(n_{p_0})$.
    Notice that for all $i = 0, 1, 2, \ldots$, we have:
    \begin{itemize}
        \item $x \notin X_{p_i}$, as $x \in X_q$ already
        \item $x \in \vars(n_{p_i})$, as $x$ remains from $n_{p_0}$ without being projected according to $\pi(n_{p_i}) = X_{p_i}$
        \item $p_i < p_{i + 1} = \chosenCluster(n_{p_i}) \le q$ by Condition \ref{cond3} of Theorem \ref{thm_csp_jt}, as $x \in \pi(n_q) \cap \vars(n_{p_{i+1}}) = X_q \cap \vars(n_{p_{i+1}}) \ne \emptyset$
    \end{itemize}
    So the strictly increasing sequence $\seq{p_i}_i$ must contain $q$, contradiction.
\qed
\end{proof}

\begin{lemma}
\label{lemma_earlier_cluster}
    Let $1 \le p, q \le m$.
    Assume $\vars(\Gamma_p) \cap X_q \ne \emptyset$.
    Then $p \le q$.
\end{lemma}
\begin{proof}
    To the contrary, assume $p > q$.
    Then by construction, $X_q \cap \vars(\Gamma_p) = \emptyset$, contradiction.
\qed
\end{proof}

\begin{lemma}[Property \ref{prop2}]
\label{lemma_prop2}
    Let $1 \le q \le m$ and variable $x \in \pi(n_q)$.
    Take an arbitrary clause $c \in \phi$ \st{} $x \in \vars(c)$.
    Then the leaf $\gamma^{-1}(c)$ is a descendant of $n_q$.
\end{lemma}
\begin{proof}
    Notice that $c \in \Gamma_p$ for some $1 \le p \le m$.
    Then $x \in \vars(c) \subseteq \vars(\Gamma_p)$.
    Note that $x \in \pi(n_q) = X_q$.
    Thus $p \le q$ by Lemma \ref{lemma_earlier_cluster}.

    Let $l = \gamma^{-1}(c)$.
    Notice that $l \in \kappa_p$ (as $c \in \Gamma_p$).
    So $l$ is a child of the node $n_p$.
    \begin{itemize}
        \item If $p = q$, then $l$ is a child of $n_q$, and we are done.
        \item If $p < q$, by Lemma \ref{lemma_internal_descedant}, we know $n_p$ is a descendant of $n_q$, as $x \in \vars(\Gamma_p) \cap \pi(n_q) = \vars(\Gamma_p) \cap X_q \ne \emptyset$.
        Then $l$ is a descendant of $n_q$.
    \end{itemize}
\qed
\end{proof}

Now we can prove Theorem \ref{thm_csp_jt}
\begin{proof}
    Algorithm \ref{alg_csp_jt} returns a tree $T$ with root $n_m$, bijection $\gamma : \Lv T \to \phi$, and labeling function $\pi : \V T \setminus \Lv T \to 2^X$.
    The project-join tree properties are satisfied, by Lemma \ref{lemma_prop1} and Lemma \ref{lemma_prop2}
\qed
\end{proof}


\subsection{Proof of Theorem \ref{thm_td_to_join}}

Let $(T, r, \gamma, \pi)$ be the object returned by Algorithm \ref{alg_td_to_join}.
We first observe that $T$ is indeed a tree with root $r$.
For each node $a \in \V{T}$, let $O(a) \in \V{S}$ denote the node in $\V{S}$ \st{} $a$ was created in the $\func{Process}(O(a), \ell)$ call for some $\ell \subseteq X$.
Throughout, let $s$ denote the value obtained on Line \ref{line_arbitrary_node} of Algorithm \ref{alg_td_to_join}.

We begin by stating three basic properties of Algorithm \ref{alg_td_to_join}.
\begin{lemma}
\label{lem:td-label}
    For each $a \in \V{T} \setminus \Lv{T}$, we have $\pi(a) \subseteq \chi(O(a))$.
\end{lemma}
\begin{proof}
    This follows from Line \ref{line_return_singleton} of Algorithm \ref{alg_td_to_join}.
\qed
\end{proof}

\begin{lemma}
\label{lem:td-parent}
    For each $a \in \V{T} \setminus \Lv{T}$ where $O(a) \neq s$, let $p$ be the parent of $O(a)$ in $S$.
    Then $\pi(a) \cap \chi(p) = \emptyset$.
\end{lemma}
\begin{proof}
    Observe that $\ell = \chi(p)$ in the $\func{Process}(O(a), \ell)$ call on Line \ref{line_recur}.
    The result then follows from Line \ref{line_return_singleton} of Algorithm \ref{alg_td_to_join}.
\qed
\end{proof}

\begin{lemma}
\label{lem:td-bounds}
    Let $n \in \V{S}$.
    For every $\ell \subseteq X$ and $i \in \func{Process}(n, \ell)$, we have $\vars(i) \subseteq \ell$.
\end{lemma}
\begin{proof}
    We proceed by induction on the tree structure of $S$.

    Let $A$ denote the set $children$ after Line \ref{line_recur} occurs and let $a \in A$.
    We first prove that $\vars(a) \subseteq \chi(n)$.
    First, assume that $a$ is a leaf node corresponding to some $c \in clauses$.
    In this case, $\vars(a) = \vars(c) \subseteq \chi(n)$ by Line \ref{line_clauses}.
    Otherwise $a \in \func{Process}(o, \chi(n))$ for some $o \in C(n)$.
    In this case, notice that $n$ is an internal node, so by the inductive hypothesis, $\vars(a) \subseteq \chi(n)$.

    Now, if $A = \emptyset$, then $\func{Process}(n, \ell)$ returns $\emptyset$, so the lemma is vacuously true.
    If $\chi(n) \subseteq \ell$, then $A$ is returned by $\func{Process}(n, \ell)$.
    So for every $i \in A$, we have $\vars(i) \subseteq \chi(n) \subseteq \ell$.

    Otherwise, $A \neq \emptyset$ and $\chi(n) \not\subseteq \ell$.
    In this case, $\func{Process}(n, \ell)$ returns a single node $i$ with $\vars(i) = \cup_{a \in A} \vars(a) \setminus (\chi(n) \setminus \ell) \subseteq \chi(n) \setminus (\chi(n) \setminus \ell) \subseteq \ell$.
\qed
\end{proof}







Given these three properties, it is straightforward to prove that $(T, r, \gamma, \pi)$ satisfies all conditions to be a project-join tree of $\phi$.
We prove each condition in a separate lemma here.

\begin{lemma}
    $\gamma$ is a bijection.
\end{lemma}
\begin{proof}
    Note that $\gamma$ is an injection since $found$ on Line \ref{line_found} of Algorithm \ref{alg_td_to_join} ensures that we generate at most one leaf node for each clause.
    To show that $\gamma$ is a surjection, consider $c \in \phi$.
    Then $\vars(c)$ forms a clique in the Gaifman graph of $\phi$.
    It follows (since the treewidth of a complete graph on $k$ vertices is $k-1$) that $\vars(c) \subseteq \chi(n)$ for some $n \in \V{S}$.
    Thus $\gamma$ is a surjection as well.
\qed
\end{proof}
\begin{lemma}
    $P = \{\pi(a) : a \in \V{T} \setminus \Lv{T} \}$ is a partition of $X$.
\end{lemma}
\begin{proof}
    First, let $x \in X$.
    Then $x \in \vars(c)$ for some $c \in \phi$.
    Since $\gamma$ is a bijection, $x \in \vars(p)$ for some $p = \gamma^{-1}(c) \in \Lv{T}$.
    However, by Lemma \ref{lem:td-bounds}, we know $x \notin \vars(r) = \emptyset$.
    Thus $x$ must have been projected out at some node $q \in \V{T}$ between $p$ and $r$.
    It follows that $x \in \pi(q) \subseteq P$.

    On the other hand, assume for the sake of a contradiction that there are distinct $a, b \in \V{T}$ \st{} $x \in \pi(a) \cap \pi(b)$.
    By Lemma \ref{lem:td-label}, $x \in \chi(O(a)) \cap \chi(O(b))$.
    Since $S$ is a tree, there is some node $p \in \V{S}$ on the path between $O(a)$ and $O(b)$ \st{} $p$ is the parent of either $O(a)$ or $O(b)$.
    By Property \ref{prop_running_intersection} of tree decompositions, $x \in \chi(p)$.
    However, this contradicts Lemma \ref{lem:td-parent}.
\qed
\end{proof}

\begin{lemma}
    For each internal node $a \in \V{T} \setminus \Lv{T}$, variable $x \in \pi(a)$, and clause $c \in \phi$ \st{} $x$ appears in $c$, the leaf node $\gamma^{-1}(c)$ is a descendant of $a$ in $T$.
\end{lemma}
\begin{proof}
    If $O(a) = s$, then $a$ is the root of $T$, so all leaf nodes are descendants.
    Otherwise, assume for the sake of contradiction that $\gamma^{-1}(c)$ is not a descendant of $a$ in $T$.
    Then $O(\gamma^{-1}(c))$ is not a descendant of $O(a)$ in $S$.
    This means that the parent $p \in \V{S}$ of $O(a)$ is on the path between $O(a)$ and $O(\gamma^{-1}(c))$.
    By Lemma \ref{lem:td-label}, we must have $x \in \chi(O(a)) \cap \chi(O(\gamma^{-1}(c)))$.
    By Property \ref{prop_running_intersection} of tree decompositions, $x \in \chi(p)$.
    But this contradicts Lemma \ref{lem:td-parent}.
\qed
\end{proof}

It follows that $(T, r, \gamma, \pi)$ is a project-join tree of $\phi$.


\subsection{Proof of Theorem \ref{thm_join_to_td}}

\begin{proof}
    Let $X = \vars(\phi)$.
    Define $\chi: \V{T} \to 2^X$ by, for all $n \in \V{T}$,
    $\chi(n) \equiv\vars(n)$ if $n \in \Lv{T}$ and $\chi(n) \equiv \vars(n) \cup \pi(n)$ otherwise.
    Then $(T, \chi)$ is a tree decomposition of the Gaifman graph of $\phi$.
    Moreover, the width of $(T, \chi)$ is $\func{tw}(T, \chi) = \max_{n \in \V{T}} \size{\chi (n)}  - 1 = \max_{n \in \V T} \func{size}(n) - 1 = w - 1$.
\qed
\end{proof}



\end{document}